\documentclass[letter,11pt]{elsarticle}

\usepackage{graphicx}
\usepackage{bm}
\usepackage{amsmath}
\usepackage{amsfonts}
\usepackage{amssymb}
\usepackage{amsthm}
\usepackage{float}
\usepackage{color}

\newcommand{\E}{\mathbb E}

\newcommand{\A}{\mathcal{A}}
\newcommand{\B}{\mathcal{B}}

\newcommand{\M}{\mathcal{M}}
\newcommand{\N}{\mathcal{N}}

\newcommand{\tb}[1]{\mathbf{#1}}

\def\x{\mathbf{x}}
\def\bx{\mathbf{x}}

\def\n{\mathbf{n}}
\def\s{\mathbf{s}}

\def\y{\mathbf{y}}

\def\0{\mathbf{0}}

\def\C{\mathbf{C}}

\def\I{\mathbf{I}}

\def\R{\mathbf{R}}

\def\1{\mathbf{1}}

\newtheorem{lem}{Lemma}
\newtheorem{thm}{Theorem}

\def\x{{\mathbf x}}

\newcommand{\lk}[1]{{\color{black}{#1}}}

\begin{document}

\begin{frontmatter}

\title{Measurement Design for Detecting Sparse Signals\tnoteref{supp}}
\tnotetext[supp]{A preliminary version of parts of this work was presented at the 2010 American Control Conference, Baltimore, MD. This work is supported in part by ONR Grant N00014-08-1-110 and NSF Grants
ECCS-0700559, CCF-0916314, and CCF-1018472.}

\author{Ramin Zahedi}
\ead{Ramin.Zahedi@colostate.edu}
\author{Ali Pezeshki}
\ead{Ali.Pezeshki@colostate.edu}
\author{Edwin K. P. Chong}
\ead{Edwin.Chong@colostate.edu}

\address{Department of Electrical and Computer Engineering, Colorado State University\\ Fort Collins, CO 80525-1373, USA}

\begin{abstract}

We consider the problem of testing for the presence (or detection) of an unknown sparse signal in additive white noise. Given a fixed measurement budget, much smaller than the dimension of the signal, we consider the general problem of designing compressive measurements to maximize the measurement signal-to-noise ratio (SNR), as increasing SNR improves the detection performance in a large class of detectors. We use a {\em lexicographic optimization} approach, where the optimal measurement design for sparsity level $k$ is sought only among the set of measurement matrices that satisfy the optimality conditions for sparsity level $k-1$. We consider optimizing two different SNR criteria, namely a worst-case SNR measure, over all possible realizations of a $k$-sparse signal, and an average SNR measure with respect to a uniform distribution on the locations of the up to $k$ nonzero entries in the signal. We establish connections between these two criteria and certain classes of tight frames. We constrain our measurement matrices to the class of tight frames to avoid coloring the noise covariance matrix. For the worst-case problem, we show that the optimal measurement matrix is a Grassmannian line packing for most---and a uniform tight frame for all---sparse signals. For the average SNR problem, we prove that the optimal measurement matrix is a uniform tight frame with minimum \emph{sum-coherence} for most---and a tight frame for all---sparse signals.

\end{abstract}


\end{frontmatter}


\section{Introduction}\label{sec0}

Over the past few years, considerable progress has been made towards developing a mathematical framework for \emph{reconstructing} sparse or compressible signals.\footnote{A vector $\bx=[x_1,x_2,\ldots,x_N]^T$ is sparse when the cardinality of its support $S=\{k: x_k\neq0\}$ is much smaller than its dimension $N$. A vector $\bx$ is compressible, if its entries obey a power law, i.e., the $k$th largest entry in absolute value, denoted by $|x|_{(k)}$, satisfies $|x|_{(k)} \leq C_r \cdot k^{-r}$, $r>1$ and $C_r$ is a constant depending only on $r$ (see, e.g., \cite{CT05}). Then $\|\bx-\bx_k\|_1 \leq \sqrt{k}C_r'\cdot k^{-r+1}$, where $\bx_k$ is the best $k$-term approximation of $\bx$.} The most notable result is the development of the compressed sensing theory (see, e.g., \cite{CT05}--\nocite{Donoho2006,Baraniuk-SPM07,candes2006stable,Romberg-SPM08}\cite{C06}),  which shows that an unknown signal can be recovered from a small (relative to its dimension) number of linear measurements provided that the signal is sparse. Thus, compressed sensing and related sparse recovery methods have become topics of great interest, leading to many exciting developments in sparse representation theory, measurement design, and sparse recovery algorithms (see, e.g, \cite{elad2010sparse}--\nocite{Rauhut-IT08,Tropp-CoSaMP,howard2008fast,applebaum2009chirp,devore2007deterministic,tropp2007signal,daubechies2010iteratively}\cite{sarvotham2006sudocodes}).

The major part of the effort, however, has been focused on {\em estimating} sparse signals. Hypothesis testing (detection and classification) involving sparse signal models, on the other hand, has been scarcely addressed, notable exceptions being \cite{daven}--\nocite{daven2010,hau,wang}\cite{paredes}. Detecting a sparse signal in noise is fundamentally different from reconstructing a sparse signal, as the objective in detection often is to maximize the probability of detection or to minimize a Bayes risk, rather than to find the sparsest signal that satisfies a linear observation equation. We note that in the compressed sensing literature the term ``sparse signal detection'' often means identifying the support of a sparse signal. In this paper, however, we use this term to refer to a binary hypothesis test for the presence or absence of a sparse signal in noise. The problem is to decide whether a measurement vector is a realization from a hypothesized noise only model or from a hypothesized signal-plus-noise model, where in the latter model the signal is sparse in a known basis but the indices and values of its nonzero coordinates are unknown.

Existing work (e.g., see \cite{daven}--\nocite{daven2010}\cite{hau}) is mainly focused on understanding how the performance of well-known detectors (e.g., the Neyman-Pearson detector) are affected by measurement matrices that have the so-called restricted isometry property (RIP). The RIP condition for the measurement matrix is sufficient for the minimum $\ell_1$-norm solution to be exact (or near-exact when the measurements are noisy) (e.g., see \cite{candes2006stable}). A fundamental result of compressed sensing has been to establish that random matrices, with independently and identically distributed (i.i.d.)\ Gaussian or i.i.d.\ Bernoulli entries, satisfy the RIP condition with high probability. The analysis presented in \cite{daven} and \cite{daven2010} provides theoretical bounds on the performance of a Neyman-Pearson detector---quantified by the maximum probability of detection achieved at a pre-specified false alarm rate---when matrices with i.i.d.\ Gaussian entries are used for collecting measurements. In~\cite{hau}, the authors derive bounds on the total error probability for detection, involving both false alarm and miss detection probabilities, but again for measurement matrices with i.i.d.\ entries. Finally, in~\cite{wang} and~\cite{paredes}, the authors develop compressive matched subspace detectors that also use random matrices for collecting measurements for detecting sparse signals in known subspaces.

The body of work reported in \cite{daven}--\nocite{daven2010,hau,wang}\cite{paredes} provides a valuable analysis of the performance of different detectors, but leave the question of how to design measurement matrices to optimize a measure of detection performance open. As in the case of reconstruction, random matrices have been studied in these papers in the context of signal detection primarily because of the tractability of the associated performance analysis. But what are the necessary and sufficient conditions a compressive measurement matrix must have to optimize a desired measure of detection performance? How can matrices that satisfy such conditions be constructed? Our aim in this paper is to take initial but significant steps towards answering these questions. We further clarify our goals and contributions in the next section.


\section{Problem Statement and Main Contributions}

We consider the design of low dimensional (compressive) measurement matrices, with a pre-specified number of measurements, for detecting sparse signals in additive white Gaussian noise. More specifically, we consider the following binary hypothesis test:
\begin{equation}\label{eq16}
\left\{\begin{array}{l}
\mathcal{H}_{0}: \tb{x}=\tb{n},\\
\mathcal{H}_{1}: \tb{x}=\tb{s}+\tb{n},
\end{array}\right.
\end{equation}
where $\x$ is an $(N\times 1)$ vector that describes the state of a physical phenomenon. Under the null hypothesis $\mathcal{H}_{0}$, $\x$ is a white Gaussian noise vector with covariance matrix $E[\n\n^H]=(\sigma_n^2/N)\I$. Under the alternative hypothesis $\mathcal{H}_{1}$, $\x=\s+\n$ consists of a deterministic signal $\s$ distorted by additive white Gaussian noise $\n$.

We assume $\s$ is $k$-sparse in a known basis ${\bm \Psi}$. That is, to say, $\s$ is composed as
\begin{equation}\label{eq9}
\s={\bm \Psi} {\bm \theta},
\end{equation}
where ${\bm \Psi}\in \mathbb{R}^{N\times N}$ is a known matrix, whose columns form an orthonormal basis for $\mathbb{R}^N$, and ${\bm \theta} \in \mathbb{R}^N$ is a $k$-sparse vector, i.e., it has between 1 to $k\ll N$ nonzero entries. We may refer to $\s$ as simply $k$-sparse for brevity.

We wish to decide between the two hypotheses based on a given number $m\le N$ of linear measurements $\y={\bm \Phi}\x$ from $\x$, where
${\bm \Phi} \in \mathbb{R}^{m\times N}$ is a compressive measurement matrix that we will design. The observation vector $\tb{y}={\bm
\Phi}\x$ belongs to one of the following hypothesized models:
\begin{equation}\label{eq:chk}
\left\{\begin{array}{l}
\mathcal{H}_{0}: \tb{y}={\bm \Phi} \tb{n}\sim \N(\0,(\sigma_n^2/N){\bm \Phi}{\bm \Phi}^{H}),\\
\mathcal{H}_{1}: \tb{y}={\bm \Phi} (\tb{s}+\tb{n})\sim \N({\bm
\Phi}\s,(\sigma_n^2/N){\bm \Phi}{\bm \Phi}^{H}),
\end{array}\right.
\end{equation}
where the superscript $H$ is the Hermitian transpose. To avoid coloring the noise vector $\n$, we constraint the compressive
measurement matrix ${\bm \Phi}$ to be right orthogonal, that is we
force ${\bm \Phi}{\bm \Phi}^H=\I$.

Rather than limiting ourselves to a particular detector, we look at the general problem of designing compressive measurements to maximize the measurement signal-to-noise ratio (SNR), under $\mathcal{H}_1$, which is given by
\begin{equation}\label{eq:SNR}
\textrm{SNR}=(\tb{s}^{H}{\bm \Phi}^{H} {\bm \Phi}
\tb{s})/(\sigma_{n}^{2}/N).
\end{equation}
This is motivated by the fact that for the class of linear log-likelihood ratio detectors, where the log-likelihood ratio is a linear function of the data, the detection performance is improved by increasing SNR. In particular, for a Neyman-Pearson detector, with false alarm rate $\gamma$, the probability of detection $P_d=Q(Q^{-1}(\gamma)-\sqrt{\textrm{SNR}})$ is monotonically increasing in SNR, where $Q(\cdot)$ is the $Q$-function. In addition, maximizing SNR leads to maximum detection probability, at a pre-specified false alarm rate, when an energy detector is used. Without loss of generality, throughout the paper we assume that $\sigma_{n}^{2}=1$ and $\|\tb{s}\|^2=\|{\bm \theta}\|^2=1$, and so we design ${\bm \Phi}$ to maximize the measured signal energy $\|{\bm \Phi}\tb{s}\|^2$.

In solving the problem, one approach is to assume a value for the sparsity level $k$ and design the measurement matrix ${\bm \Phi}$ based on this assumption. This approach, however, runs the risk that the true sparsity level might be different. An alternative approach is not to assume any specific sparsity level. Instead, when designing the measurement matrix ${\bm \Phi}$, we prioritize the level of importance of different values of sparsity $k$. In other words, we first find a set of solutions that are optimal for a $k_{1}$-sparse signal. Then, within this set, we find a subset of solutions that are also optimal for $k_{2}$-sparse signals. We follow this procedure until we find a subset that contains a family of optimal solutions for sparsity levels $k_{1}$, $k_{2}$, $k_{3}$, $\cdots$. This approach is known as a \emph{lexicographic optimization} method (see, e.g., \cite{Iser}--\nocite{hajek}\cite{ehr}).

Replacing \eqref{eq9} in \eqref{eq:SNR} yields
\[
\textrm{SNR}=\frac{\|{\bm \Phi}{\bm \Psi} {\bm \theta} \|^{2}}{(\sigma_{n}^2/N)}.
\]
The basis matrix ${\bm\Psi}$ is known, but the $k$-sparse representation vector ${\bm \theta}$ is unknown. That is, the exact number of the nonzero entries in ${\bm \theta}$, their locations, and their values are unknown. The measurement design naturally depends on one's assumptions about the unknown vector ${\bm \theta}$. We consider two different design problems, namely a worst-case SNR design and an average SNR design, as explained below.

\emph{Worst-case SNR design.} In the first case, we assume the vector ${\bm \theta}$ is deterministic but unknown. Then, among all possible deterministic $k$-sparse vectors ${\bm \theta}$, we consider the vector that minimizes the SNR and design the matrix ${\bm \Phi}$ that maximizes this minimum SNR. Of course, when minimizing the SNR with respect to ${\bm \theta}$, we have to find the minimum SNR with respect to locations and values of the nonzero entries in the vector ${\bm \theta}$. To combine this with the lexicographic approach, we design the matrix ${\bm \Phi}$ to maximize the \emph{worst-case} detection SNR, where the worst-case is taken over all subsets of size $k_i$ of elements of ${\bm \theta}$, where $k_i$ is the sparsity level considered at the $i$th level of lexicographic optimization. This is a design for robustness with respect to the worst sparse signal that can be produced in the basis ${\bm \Psi}$. The reader is referred to Section \ref{sec1} for a complete statement of the problem.

We show (see Section \ref{sec2}) that the worst-case detection SNR is maximized when the columns of the product ${\bm \Phi}{\bm \Psi}$ between the compressive measurement matrix ${\bm \Phi}$ and the sparsity basis ${\bm \Psi}$ form a \emph{uniform tight frame}. A uniform tight frame is a frame system in which the frame operator is a scalar multiple of the identity operator and every frame element has the same norm (see, e.g., \cite{casazza2006existence}). We also show that when the signal is $2$-sparse, the optimal frame is a \emph{Grassmannian line packing} (see, e.g., \cite{conway}). For the case where the sparsity level of the signal is greater than two, we provide a lower bound on the worst-case performance. If the number $m$ of measurements allowed is greater than or equal to $\sqrt{N}$, then the Grassmannian line packing frame will be an equiangular uniform tight frame (see, e.g., \cite{stro}--\nocite{pez,SH03,Bodmann,CK03,R07,STDH07}\cite{mal}) and the maximal worst-case SNR can be expressed in terms of the Welch bound. Numerical examples presented in Section \ref{sec5} show that Grassmannian line packing frames provide better worst-case performance than matrices with i.i.d.\ Gaussian entries, which are typically used in sparse signal reconstruction.

\emph{Average SNR design.} In the second case, we assume that the locations of nonzero entries of ${\bm \theta}$ are random but their values are deterministic and unknown. We find the matrix ${\bm \Phi}$ that maximizes the expected value of the minimum SNR. The expectation is taken with respect to a random index set with uniform distribution over the set of all possible subsets of size $k_i$ of the index set $\{1,2,\ldots,N\}$ of elements of ${\bm \theta}$. The minimum SNR, whose expected value we wish to maximize, is calculated with respect to the values of the entries of the vector ${\bm\theta}$ for each realization of the random index set. The reader is referred to Section \ref{sec3} for a complete statement of the problem.

We show (see Section \ref{sec4}) that for 1-sparse signals, any right orthogonal measurement matrix ${\bm \Phi}$, i.e., any tight frame, is optimal for maximizing the average minimum SNR. For signals with sparsity levels higher than one, we constrain ourselves to the class of uniform tight frames and show that optimal measurement matrix is a uniform tight frame that has minimal \emph{sum-coherence}, as described in Section \ref{sec4}. However, to the best of our knowledge constructing such frames remains an open problem in frame theory. Therefore, we limit ourselves to providing performance bounds in the average-case problem.


\section{The Worst-case Problem Statement}\label{sec1}

Since all sparse signals share the fact that they have at least one nonzero entry, it seems natural to first find an optimal measurement matrix for 1-sparse signals. Next, among the set of optimal solutions for this case, we find matrices that are optimal for 2-sparse signals. This procedure is continued for signals with higher sparsity levels. This is a lexicographic optimization approach to maximizing the worst-case SNR.

Consider the $k$th step of the lexicographic approach. In this step, the vector ${\bm \theta}$ has up to $k$ nonzero entries. We do not impose any prior constraints on the locations and the values of the nonzero entries of ${\bm \theta}$. As mentioned earlier, we assume that $\|\tb{s}\|^2=\|{\bm
\theta}\|^2=1$ and $\sigma_n^2=1$. We wish to maximize the minimum (worst-case) SNR, produced by assigning the worst possible locations and values to the nonzero entries of the $k$-sparse vector ${\bm \theta}$. Referring to \eqref{eq:SNR}, this is a worst-case design for maximizing the signal energy $\s^H{\bm \Phi}^H {\bm \Phi}\s$ inside the subspace $\langle {\bm \Phi}^H\rangle$ spanned by the columns of ${\bm \Phi}^H$, since ${\bm \Phi}^H{\bm \Phi}$ is the orthogonal projection operator onto $\langle {\bm \Phi}^H\rangle$.

To define the $k$th step of the optimization procedure more precisely, we need some additional notation. Let $\A_{0}$ be the set containing all $(m \times N)$ right orthogonal matrices ${\bm \Phi}$. Then, we recursively define the set $\A_k$, $k=1,2,\ldots\,$, as the set of solutions to the following optimization problem:

\begin{equation}\label{eq5}
\begin{array}{cl}
\mathop{\max}\limits_{{\bm \Phi}} \mathop{\min}\limits_{\tb{s}} & \|{\bm \Phi}\tb{s}\|^{2},\\
\textrm{s.t.} & {\bm \Phi} \in \A_{k-1},\\
&\|\tb{s}\|=1,\\
& \mbox{$\tb{s}$ is $k$-sparse.}
\end{array}
\end{equation}
In our lexicographic formulation, the optimization problem for the $k$th problem~\eqref{eq5} involves a worst-case objective restricted to the set of solutions $\A_{k-1}$ from the $(k-1)$th problem. So, $\A_{k}\subset \A_{k-1}\subset \cdots \subset \A_{0}$.

Before we present a complete solution to these problems, we first simplify them in three steps. First, since the matrix ${\bm \Psi}$ is known, the matrix ${\bm \Phi}$ can be written as ${\bm \Phi}=\C {\bm \Psi}^H$, where $\C$ is an $(m \times N)$ matrix. Then, ${\bm \Phi} {\bm \Psi}=\C {\bm \Psi}^{H} {\bm \Psi}=\C$, and also ${\bm \Phi}{\bm \Phi}^H=\C{\bm \Psi}^{H}{\bm \Psi} \C^H =\C\C^H=\I$. Using \eqref{eq9}, the max-min problems~(\ref{eq5}) become
\begin{equation}\label{eq6}
\begin{array}{cl}
\mathop{\max}\limits_{\C} \mathop{\min}\limits_{{\bm \theta}} & \|\C {\bm \theta}\|^{2},\\
\textrm{s.t.} &\C\in \B_{k-1},\\
& \|{\bm \theta}\|=1,\\
& \mbox{$\bm\theta$ is $k$-sparse,}
\end{array}
\end{equation}
where $\B_{0}=\A_{0}$, and similar to the sets $\A_{k}$, the sets $\B_{k}$ ($k=1,2,\ldots\,$) are recursively defined to contain all the optimal solutions of~(\ref{eq6}). It is easy to see that $\B_{k}=\{\C: \C{\bm \Psi}^H\in\A_{k}\}$.

Let $\Omega=\{1,2,\ldots, N\}$ and define $\Omega_{k}$ to be $\Omega_{k}=\{E \subset \Omega: |E|=k\}$. For any $T \in \Omega_{k}$, let ${\bm \theta}_{T}$ be the subvector of size $(k \times 1)$ that contains all the components of ${\bm \theta}$ corresponding to indices in $T$. Similarly, given a matrix $\C$, let $\C_{T}$ be the $(m \times k)$ submatrix consisting of all columns of $\C$ whose indices are in $T$. Note that the vector ${\bm \theta}_{T}$ may have zero entries. Indeed, for cases where the $k$-sparse vector ${\bm \theta}$ has fewer than $k$, e.g., $l<k$, nonzero entries, the $(k \times 1)$ vector ${\bm \theta}_{T}$ has $k-l$ zero entries. This is important because our definition for $T$ and ${\bm \theta}_{T}$ is slightly different than the common definitions used in the compressed sensing literature, where $T$ and ${\bm \theta}_{T}$ only contain indices and values related to the nonzero entries of the vector ${\bm \theta}$, often called the support of $T$. We refer to a member $T$ of $\Omega_{k}$ as a ``$k$-platform''. Thus, a $k$-platform $T$ includes, but is not limited to, the support of the sparse vector ${\bm \theta}$.

Given $T\in \Omega_{k}$, the product $\C{\bm \theta}$ can be replaced by $\C_{T} {\bm \theta_{T}}$ instead. Now, to consider the worst-case scenario for the SNR, as well as considering the worst ${\bm \theta}_{T}$ that minimizes $\|\C_{T}{\bm \theta}_{T}\|^{2}$, we also have to consider the worst $T\in \Omega_{k}$. Thus, the max-min problem becomes
\begin{equation}\label{eq:mmm}
\begin{array}{cl}
\mathop{\max}\limits_{\C} \mathop{\min}\limits_{T} \mathop{\min}\limits_{{\bm \theta}_{T}}& \|\C_{T} {\bm \theta}_{T}\|^{2},\\
\textrm{s.t.} & \C\in \B_{k-1},\\
& \|{\bm \theta}_{T}\|=1, T \in \Omega_{k}.
\end{array}
\end{equation}
The solution to \eqref{eq:mmm} is the most robust design with respect to the locations and values of the nonzero entries of the parameter vector ${\bm \theta}$.

The solution to the minimization subproblem
\[
\begin{array}{cl}
\mathop{\min}\limits_{{\bm \theta}_{T}}& \|\C_{T} {\bm \theta}_{T}\|^{2},\\
\textrm{s.t.} & \|{\bm \theta}_{T}\|=1,
\end{array}
\]
is well known; see, e.g., \cite{chong}. The optimal objective function is $\lambda_{\textrm{min}}(\C^H_{T}\C_{T})$, the smallest eigenvalue of the matrix $\C^H_{T}\C_{T}$. Therefore, the max-min-min problem \eqref{eq:mmm} simplifies to
\begin{equation}\label{eq1}
(\textrm{P}_{k}) \quad
\left\{\begin{array}{cl}
\mathop{\max}\limits_{\C} \mathop{\min}\limits_{T} & \lambda_{\textrm{min}}(\C^H_{T}\C_{T}),\\
\textrm{s.t.} & \C \in \B_{k-1},\\
& T \in \Omega_{k}.
\end{array}\right.
\end{equation}
At each step $k$, the optimal compressive measurement matrix, denoted by ${\bm \Phi}^*$, is determined from the optimizer $\C^*$ of \eqref{eq1}
as ${\bm \Phi}^*=\C^*{\bm \Psi^H}$.

Next, we describe how to solve the max-min problem $(\textrm{P}_{k})$ in \eqref{eq1}.


\section{Solution to the Worst-case Problem}\label{sec2}

Let $\tb{c}_{i}$ be the $i$th column of the matrix $\C$. We first find the solution set $\A_{1}$ for problem $(\textrm{P}_{1})$. Then, we find a subset $\A_{2}\subset \A_{1}$ as the solution for $(\textrm{P}_{2})$. We continue this procedure for general sparsity level $k$.

\subsection{Sparsity Level $k=1$}

If $k=1$, then any $T$ such that $|T|=1$ can be written as $T=\{ i \}$ with $i\in \Omega$, and $\C_{T}=\tb{c}_i$ consists of only the $i$th column of $\C$. Therefore, $\C^H_{T}\C_{T}=\tb{c}_{i}^{H}\tb{c}_{i}=\|\tb{c}_{i}\|^{2}$, and the max-min problem becomes
\begin{equation}\label{eq7}
\begin{array}{cl}
\mathop{\max}\limits_{\C} \mathop{\min}\limits_{i} &
\|\tb{c}_{i}\|^{2},\\
\textrm{s.t.} & \C \in \B_{0}, \\
& i \in \Omega.
\end{array}
\end{equation}
\begin{thm}\label{thm1}
The optimal value of the objective function of the max-min problem~(\ref{eq7}) is $m/N$. A necessary and sufficient condition for a matrix $\C^{*}$ to be in the solution set $\B_{1}$ is that the columns $\{\tb{c}^{*}_i\}_{i=1}^N$ of $\C$ form a {\em uniform tight frame} with norm values equal to $\sqrt{m/N}$.
\end{thm}
\begin{proof} We first prove the claim about the optimal value. Assume false, i.e., assume there exists an optimal matrix $\C^{*}\in \B_{1}$ for which the value of the cost function is either less than or greater than $m/N$. Suppose the former is true. Let $\C_{1}$ be an $(m \times N)$ matrix,
satisfying $\C_{1}\C_{1}^H=\I$, whose columns have equal norm $\sqrt{m/N}$. Then, the value of the objective function in~\eqref{eq7} for $\C=\C_{1}$ is $m/N$. This means that our proposed matrix $\C_{1}$ achieves a higher SNR than $\C^{*}$ which is a contradiction. Now, assume the latter is
correct, that is the value of the objective function for $\C^{*}$ is greater than $m/N$. This means $\mathop{\min}\limits_{i \in \Omega}
\|\tb{c}_{i}^*\|^{2}=\|\tb{c}_{j}^*\|^{2}>m/N$. Knowing this, we write
\[
\textrm{tr}\left(\C^*\C^{*H}\right)=\textrm{tr}\left(\C^{*H}\C^{*}\right)=\sum\limits_{i=1}^N \|\tb{c}_{i}^*\|^{2}
>\sum\limits_{i=1}^N m/N=m.
\]
However, from the constraint in \eqref{eq7} we know that $\C^{*}\C^{*H}=\I$, and $\textrm{tr}(\C^{*}\C^{*H})=m$. This is also a contradiction. Thus, the assumption is false and the optimal value for the objective function of \eqref{eq7} is $m/N$.

We now prove the claim about the optimizer $\C^{*}$. From the preceding part of the proof, it is easy to see that all columns of $\C^{*}$ must have equal norm $\sqrt{m/N}$. If not, since none of them can be less than $\sqrt{m/N}$, then the sum of all column norms will be greater than $m$, which is a contradiction. Moreover, we write
\begin{equation}\label{eq8}
\C^{*}\C^{*H}=\sum_{i=1}^{N}\tb{c}_{i}^*\tb{c}_{i}^{*H}=\I.
\end{equation}
Multiplying both sides of \eqref{eq8} by an arbitrary $(m\times 1)$ vector $\tb{x}$ from the right side and $\tb{x}^{H}$ from the left side, we get $\sum_{i=1}^{N}\|\tb{c}_{i}^{*H}\tb{x}\|^2=\|\tb{x}\|^2$. This equation represents a tight frame with frame elements $\{\tb{c}_{i}^*\}$ and frame bound $1$. In other words, it represents a Parseval frame. Since the frame elements have equal norms, the frame is also uniform. Therefore, for a matrix $\C^{*}$ to be in $\B_{1}$, the columns of $\C^{*}$ must form a \emph{uniform tight frame}.
\end{proof}

\textit{Remark 1:} The reader is referred to \cite{casazza2006existence}, \cite{Bodmann}, \cite{casazza2010constructing}, \cite{calderbanksparse}, and the references therein, for examples of constructions of uniform tight frames.

\subsection{Sparsity Level $k=2$}

The next step is to solve $(\textrm{P}_{2})$. Since our solution for this case should lie among the family of optimal solutions for $k=1$, results concluded in the previous part should also be taken into account, i.e., the columns of the optimal matrix $\C^{*}$ must form a uniform tight frame, where the frame elements $\tb{c}^{*}_i$ have norm $\sqrt{m/N}$.

Given $T\in \Omega_{2}$, the matrix $\C_{T}$ consists of two columns, e.g., $\tb{c}_{i}$ and $\tb{c}_{j}$. So, the matrix $\C_{T}^H\C_{T}$ in the
max-min problem~(\ref{eq1}) is a $(2 \times 2)$ matrix:
\[
\C_{T}^H\C_{T}=\left[\begin{array}{cc}
\langle\tb{c}_{i},\tb{c}_{i}\rangle & \langle\tb{c}_{i},\tb{c}_{j}\rangle \\
\langle\tb{c}_{i},\tb{c}_{j}\rangle &
\langle\tb{c}_{j},\tb{c}_{j}\rangle
\end{array}\right].
\]
From the $k=1$ case, we have $\|\tb{c}_i\|^2=\|\tb{c}_j\|^2=m/N$. Therefore,
\[
\C_{T}^H\C_{T}=(m/N)\left[\begin{array}{cc}
1 & \cos{\alpha_{ij}} \\
\cos{\alpha_{ij}} & 1
\end{array}\right],
\]
where $\alpha_{ij}$ is the angle between vectors $\tb{c}_{i}$ and $\tb{c}_{j}$. The minimum eigenvalue of this matrix is
\begin{equation}\label{eq19}
\lambda_{\textrm{min}}(\C_{T}^H\C_{T})=(m/N)(1-|\cos{\alpha_{ij}}|).
\end{equation}
Given any matrix $\C \in \B_{1}$, define \emph{coherence} $\mu_{\C}$ as
\begin{equation}\label{eq26}
\mu_{\C}=\mathop{\max}\limits_{\tb{c}_{i}, \tb{c}_{j}: \textrm{ columns of }\C}\frac{|\langle\tb{c}_{i},\tb{c}_{j}\rangle|}{\|\tb{c}_{i}\|\|\tb{c}_{j}\|}.
\end{equation}
Also, let $\mu^{*}$ be
\begin{equation}\label{eq27}
\mu^*=\mathop{\min}\limits_{\C \in \B_{1}}\mu_{\C}.
\end{equation}
The following theorem holds.

\begin{thm}\label{thm2}
The optimal value of the objective function of the max-min problem $(\textrm{P}_{2})$ is $(m/N)(1-\mu^*)$. A matrix $\C^{*}$ is in $\B_{2}$ if and only if the columns of $\C^{*}$ form a \emph{uniform tight frame} with norm values $\sqrt{m/N}$ and $\mu_{\C^*}=\mu^*$.
\end{thm}

\begin{proof}Since our solution must be chosen from the family of uniform tight frames with frame elements of equal norm $\sqrt{m/N}$, the objective function of $(\textrm{P}_{2})$ is only a function of the angle $\alpha_{ij}$. Using~(\ref{eq19}) and~\eqref{eq26}, it is easy to see that the minimum $\lambda_{\textrm{min}}(\C_{T}^H\C_{T})$ is $(m/N)(1-\mu_{\C})$. Using~(\ref{eq27}), we conclude that the largest possible value of the objective function of $(\textrm{P}_{2})$ is $(m/N)(1-\mu^*)$.
\end{proof}

\textit{Remark 2:} Methods for constructing uniform tight frames with frame elements that have a coherence $\mu^*$ is equivalent to optimal \emph{Grassmannian packings} of one-dimensional subspaces, or \emph{Grassmannian line packings} (see, e.g., \cite{conway}--\nocite{stro,pez,SH03,Bodmann,CK03,R07,sus}\cite{mal}). We will say more about this point later in the paper.

\textit{Remark 3:} In the case where $k=2$, the matrix $\C_{T}^H\C_{T}$ (where $\C \in \B_{1}$), for any choice of $T\in \Omega_2$, is a $(2 \times 2)$ matrix with minimum and maximum eigenvalues equal to $(m/N)(1\pm |\cos{\alpha_{ij}}|)$. Therefore, the matrix $\C_{T}^H\C_{T}$ with eigenvalues equal to $(m/N)(1\pm \mu_{\C})$ has the smallest minimum eigenvalue and the largest maximum eigenvalue among eigenvalues of all matrices of the form $\C_{T}^H\C_{T}$ (for a fixed $\C$ and a varying $T$). Moreover, among all $\C \in B_{1}$, when comparing the resulting submatrices $\C_{T}^H\C_{T}$ for $T \in \Omega_2$, the matrix $\C^*$ with coherence $\mu^*$ has the largest minimum eigenvalue $(m/N)(1- \mu^*)$ and the smallest maximum eigenvalue $(m/N)(1+ \mu^*)$. This means that given any vector $\s \in \R^2$ and $T \in \Omega_2$, the following inequalities hold:
\begin{equation}\label{wq1}
(1-\mu^*)\|\s\|^2\leq \|\C^*_{T}\s\|^2\leq (1+\mu^*)\|\s\|^2.
\end{equation}

Recall the definition of Restricted Isometry Property (RIP) (see, e.g., \cite{C06}): Let $\mathbf{A}$ be a $(p \times q)$ matrix and let $l\leq q$ be an integer. Suppose $\delta_l\ge 0$ is the smallest constant such that, for every $(p \times l)$ submatrix $\mathbf{A}_{l}$ of $\mathbf{A}$ and every $(l \times 1)$ vector $\s$,
\[
(1-\delta_{l})\|\s\|^2\leq \|\mathbf{A}_{l}\s\|^2\leq (1+\delta_{l})\|\s\|^2.
\]
Then, the matrix $\mathbf{A}$ is said to satisfy the $l$-restricted isometry property ($l$-RIP) with the restricted isometry constant (RIC) $\delta_{l}$.

By comparing the 2-RIP definition with \eqref{wq1}, we can conclude that the optimal matrix $\C^*$ not only satisfies the 2-RIP with RIC $\mu^*$, but also among all matrices that satisfy 2-RIP and have uniform column norms equal to $\sqrt{m/N}$, it provides the best RIC. Thus, our solution for  optimizing the worst-case SNR for 2-sparse signals is also the ideal matrix for recovering 2-sparse signals based on methods that rely on the RIP condition for their performance guarantees.

\subsection{Sparsity Level $k>2$}

We now consider the case where $k>2$. In this case, $T\in \Omega_{k}$ can be written as $T=\{i_{1},i_{2},\cdots,i_{k}\}\subset\Omega$. From the previous results, we know that an optimal matrix $\C^{*}\in \B_{k}$ must already satisfy two properties, in addition to $\C^{*}\C^{*H}=\I$:
\begin{itemize}
\item Columns of $\C^{*}$ must build a uniform tight frame with equal norm $\sqrt{m/N}$ (to be in the set $\B_{1}$),
\item The coherence $\mu_{\C^*}$ should be equal to $\mu^*$ (to be in the set $\B_{2}$).
\end{itemize}

Taking the above properties into account for $\C^{*}$, the matrix $\C_{T}^{*H}\C_{T}^{*}$ will be a $(k \times k)$ symmetric matrix that can be written as $\C_{T}^{*H}\C_{T}^{*}=(m/N)[\I+\mathbf{A}_{T}]$ where $\mathbf{A}_{T}$ is
\begin{equation}\label{eq29}
\mathbf{A}_{T}=\left[\begin{array}{cccc}
0 & \cos{\alpha_{i_{1}i_{2}}^*} & \ldots & \cos{\alpha_{i_{1}i_{k}}^*}\\
\cos{\alpha_{i_{1}i_{2}}^*} & 0 & \ldots & \cos{\alpha_{i_{2}i_{k}}^*}\\
\vdots & \vdots & \ddots& \vdots\\
\cos{\alpha_{i_{1}i_{k}}^*} & \cos{\alpha_{i_{2}i_{k}}^*} &\ldots & 0 \\
\end{array}\right],
\end{equation}
where $i_{h}\neq i_{f} \in T$ for the entry $\cos{\alpha_{i_{h}i_{f}}^*}$ in the $(i_{h},i_{f})$th location. Then,
\begin{equation}\label{eq22}
\lambda_{\textrm{min}}(\C_{T}^{*H}\C_{T}^{*})=(m/N)(1+\lambda_{\textrm{min}}(\mathbf{A}_{T})).
\end{equation}

So, the problem simplifies to
\begin{equation}
(\textrm{P}_{k}) \quad
\left\{\begin{array}{cl}
\mathop{\max}\limits_{\C} \mathop{\min}\limits_{T} & \lambda_{\textrm{min}}(\mathbf{A}_{T}),\\
\textrm{s.t.} & \C \in \B_{k-1},\\
& T \in \Omega_{k}.
\end{array}\right.
\end{equation}
Solving the above problem is not trivial. It is worth mentioning that, as we will discuss later, the family of frames lying in the set $\B_{2}$ are known to be Grassmannian line packings. Building such frames is known to be very hard and in fact, for a lot of values of $m$ and $N$, no solution has been found so far (see, e.g.,~\cite{conway}). This means that building solutions for problems ($\textrm{P}_{k}$) is even a harder task. Nevertheless, we provide bounds on the value of the optimal objective function.

Given $T \in \Omega_{k}$, let $\delta_{i_{h}i_{f}}^*$ be
\begin{equation}\label{eq28}
\delta_{i_{h}i_{f}}^*=\mu^*-|\cos{\alpha_{i_{h}i_{f}}^*}|, \quad i_{h}\neq i_{f} \in T.
\end{equation}
Also, define $\Delta^{*}$ in the following way:
\[
\Delta^{*}=\mathop{\min}\limits_{T \in \Omega_{k}} \sum_{i_{h}\neq i_{f} \in T} \delta_{i_{h}i_{f}}^*.
\]
The following theorem holds.
\begin{thm}\label{thm3}
The optimal value of the objective function of the max-min problem $(\textrm{P}_{k})$ for $k>2$ lies between $(m/N)(1-{k \choose 2}\mu^*+\Delta^{*})$ and $(m/N)(1-\mu^*)$.
\end{thm}

\begin{proof}Let $\tb{x}_{ij}$ and $\tb{y}_{ij}$ be two $(k \times 1)$ vectors such that $\tb{x}_{ij}$ contains values $(1/\sqrt{2})$ and $(-1/\sqrt{2})$ and $\tb{y}_{ij}$ contains values $(1/\sqrt{2})$ and $(1/\sqrt{2})$ in the $i$th and $j$th locations $(i\neq j)$ and zeros elsewhere. Then, by using Rayleigh's inequality, i.e.,
\[
\lambda_{\textrm{min}}(\mathbf{A}_{T})\leq \frac{\tb{x}^H \mathbf{A}_{T} \tb{x}}{\tb{x}^H\tb{x}},
\]
for the matrix $\mathbf{A}_{T}$ defined above and the family of vectors $\{\tb{x}_{ij}\}$ and $\{\tb{y}_{ij}\}$ defined by $i$ and $j$ (chosen from the set $\{1,2,\ldots, k\}$), we conclude that $\lambda_{\textrm{min}}(\mathbf{A}_{T})\leq -|\cos{\alpha_{i_{h}i_{f}}^*}|,\  i_{h}\neq i_{f}\in T$.
Thus,
\begin{equation}\label{eq20}
\mathop{\min}\limits_{T\in \Omega_{k}}\lambda_{\textrm{min}}(\mathbf{A}_{T})\leq \mathop{\min}\limits_{i_{h}\neq i_{f} \in T \atop T\in \Omega_{k}}(-|\cos{\alpha_{i_{h}i_{f}}^*}|)=-\mu^{*}.
\end{equation}

Given $T\in \Omega_{k}$, the matrix $\mathbf{A}_{T}$ can be written as summation of ${k \choose 2}$ matrices $\mathbf{F}_{i_{h}i_{f}}$ ($i_{h}\neq i_{f}\in T$) where each matrix $\mathbf{F}_{i_{h}i_{f}}$ has the entry $\cos{\alpha_{i_{h}i_{f}}^*}$ in the $(i_{h},i_{f})$th and $(i_{f},i_{h})$th locations and zeros elsewhere. Using matrix properties (see, e.g., \cite{kopol}), we can write
\begin{align*}
\lambda_{\textrm{min}}(\mathbf{A}_{T})&\geq \sum_{i_{h}\neq i_{f} \in T \atop T\in \Omega_{k}} \lambda_{\textrm{min}}(\mathbf{F}_{i_{h}i_{f}})=\sum_{i_{h}\neq i_{f} \in T \atop T\in \Omega_{k}} -|\cos{\alpha_{i_{h}i_{f}}^*}|\\ & =\sum_{i_{h}\neq i_{f} \in T \atop T\in \Omega_{k}}-\mu^*+\delta_{i_{h}i_{f}}^*=-{k \choose 2}\mu^*+\sum_{i_{h}\neq i_{f} \in T \atop T\in \Omega_{k}}\delta_{i_{h}i_{f}}^*.
\end{align*}
Therefore,
\begin{equation}\label{eq21}
\mathop{\min}\limits_{T\in \Omega_{k}}\lambda_{\textrm{min}}(\mathbf{A}_{T})\geq -{k \choose 2}\mu^*+\Delta^{*}.
\end{equation}
Using~(\ref{eq22}), (\ref{eq20}), and~(\ref{eq21}) we get
\setlength\arraycolsep{0.1em}
\begin{align}\label{eq23}
(m/N)(1-\mu^*) \geq \mathop{\min}\limits_{T\in \Omega_{k}}\lambda_{\textrm{min}}(\C_{T}^{*H}\C_{T}^{*})\qquad & \nonumber \\
 \geq (m/N)(1-{k \choose 2} \mu^*+\Delta^{*}).&
\end{align}
This completes the proof.
\end{proof}

\subsection{Equiangular Uniform Tight Frames and Grassmannian Packings}

The inequality~(\ref{eq23}) in Theorem~(\ref{thm3}) suggests that if all angles between column pairs are equal, then the optimal value of the objective function of $(\textrm{P}_{k})$ for $k>2$ will reach its upper bound. In this case, the columns of $\C^{*}\in \B_{k}$ in fact form an \emph{equiangular uniform tight frame}.

Equiangular uniform tight frames are Grassmannian packings, where a collection of $N$ one-dimensional subspaces are packed in $\mathbb{R}^m$ such that the chordal distance between each pair of subspaces is the same (see, e.g., \cite{conway}, \cite{pez}, and \cite{SH03}). Each one-dimensional subspace is the span of one of the frame element vectors $\tb{c}_i$. The chordal distance between the $i$th subspace $\langle \tb{c}_i \rangle$ and the $j$th subspace $\langle \tb{c}_j \rangle$ is given by
\begin{equation}\label{eq11}
d_{c}(i,j)=\sqrt{\sin^{2}{\alpha_{ij}}},
\end{equation}
where $\alpha_{ij}$ is the angle between $\tb{c}_i$ and $\tb{c}_j$. When all the $\alpha_{ij}$, $i\neq j$, are equal and the frame is tight, the chordal distances between all pairs of subspaces become equal, i.e., $d_c(i,j)=d_c$ for all $i\neq j$, and they take their maximum value. This maximum value is the simplex bound given by
\begin{equation}\label{eq10}
d_{c}=\sqrt{(N(m-1))/(m(N-1))}.
\end{equation}
Alternatively, the largest absolute value of the cosine of the angle between any two frame elements is bounded as
\[
\max\limits_{i\neq j}|\cos \alpha_{ij}|\ge \sqrt{(N-m)/(m(N-1))}.
\]
The derivation of this lower bound is originally due to Welch \cite{Welch}. The Welch bound, or alternatively the simplex bound, are reached if and only if the vectors $\{\tb{c}_i\}_{i=1}^N$ form an equiangular uniform tight frame. This is possible only for some values of $m$ and $N$. It is shown in~\cite{STDH07} that this is possible only when $1<m<N-1$ and
\begin{equation}\label{eq24}
N\leq \min\{m(m+1)/2,(N-m)(N-m+1)/2\}
\end{equation}
for frames with real elements, and
\begin{equation}\label{eq25}
N\leq \min\{m^{2},(N-m)^{2}\}
\end{equation}
for frames with complex elements. If the above conditions hold, then the optimal solution for $(\textrm{P}_{k})$ for $k>2$ is a matrix $\C^{*}$ such that its columns form an equiangular uniform tight frame with frame elements of equal norm $\sqrt{m/N}$ and angle $\alpha$ defined as
\begin{equation}\label{eq13}
\alpha=\arcsin{\left(\sqrt{\left(\frac{m-1}{m}\right)\left(\frac{N}{N-1}\right)}\right)}.
\end{equation}
The optimal value of the objective function of $(\textrm{P}_{k})$ in this case is $(m/N)(1-\mu^*)$, where $\mu^*=|\cos{\alpha}|=\sqrt{(N-m)/(m(N-1))}$.

In other cases where $N$ and $m$ do not satisfy the condition~(\ref{eq24}) or~(\ref{eq25}), the following inequality provides a tighter bound than the simplex bound for $\mu^*$ for some values of $N$ and $m$ (see~\cite{mixon}):
\[
\mu^{*}\geq \cos{\left[\pi \left(\frac{(m-1)}{N\sqrt{\pi}}\frac{\Gamma(\frac{m+1}{2})}{\Gamma(\frac{m}{2})}\right)^{1/(m-1)}\right]}.
\]
Applying the above inequalities to~(\ref{eq23}), we conclude that by using a Grassmannian \lk{line} packing where the $k$ largest angles among angles between column pairs of the matrix $\C^{*}$ are as close as possible to the angle $\alpha$ related to $\mu^*$, the value of the SNR is guaranteed to be higher than the computed lower bound. This is, however, a very difficult problem since even finding Grassmannian line packings for different values of $N$ and $m$ is still an open problem. The reader is referred to \cite{conway} and \cite{SH03} for more details.

We have thus considered a worst-case design criterion in which we assume nothing about the vector ${\bm \theta}$, and our design is robust against arbitrary possibilities of this unknown.


\section{The Average-case Problem Statement}\label{sec3}

In the worst-case problem, an optimal $k$-platform $T$ for problem $(\textrm{P}_{k})$ is a member of $\Omega_{k}$ that minimizes $\|\C_{T}{\bm \theta_{T}}\|^{2}$. In this section, instead of finding the worst-case $T$, we consider an average-case problem with a random $T$. Let $T_{k}$ to be a random variable taking values in $\Omega_{k}$, uniformly distributed over $\Omega_{k}$. In other words, if we let $p_{k}(t)$ be the probability that $T_{k}=t$ where $t \in \Omega_{k}$, then
\[
p_{k}(t)={N \choose k}^{-1},\qquad \forall t\in \Omega_{k}.
\]
Our goal is to find a measurement matrix ${\bm \Phi}$ that maximizes the expected value of the minimum SNR, where the expectation is with respect to the random $k$-platform $T_k$, and the minimum is with respect to the entries of the vector ${\bm\theta}$ on $T_k$. Taking into account the simplifying steps used earlier for the worst-case problem in Section~\ref{sec1} and also adopting the lexicographic approach, the problem of maximizing the average SNR can then be formulated in the following way: Let $\N_{0}$ be the set containing all $(m \times N)$ right orthogonal matrices. Then for $k=1,2,\ldots\,$, recursively define the set $\N_{k}$ as the solution set to the following optimization problem:
\begin{equation}
\left\{\begin{array}{cl}
\mathop{\max}\limits_{\C} \E_{T_{k}} \mathop{\min}\limits_{{\bm \theta}_{k}} & \|\C_{T_{k}}{\bm \theta}_{k}\|^{2},\\
\textrm{s.t.} & \C \in \N_{k-1},\\
&\|{\bm \theta}_{k}\|=1,
\end{array}\right.
\end{equation}
where $\E_{T_{k}}$ is the expectation with respect to $T_{k}$. As before, the $(m \times k)$ matrix $\C_{T_{k}}$ are the columns of $\C$ whose indices are in $T_{k}$. The above can be simplified to the following:
\begin{equation}\label{eq12}
(\textrm{F}_{k}) \quad
\left\{\begin{array}{cl}
\mathop{\max}\limits_{\C} & \E_{T_{k}}  \lambda_{\textrm{min}}(\C_{T_{k}}^{H}\C_{T_{k}}),\\
\textrm{s.t.} & \C \in \N_{k-1}.
\end{array}\right.
\end{equation}


\section{Solution to the Average-case Problem}\label{sec4}

To solve the lexicographic problems $(\textrm{F}_{k})$, we follow the same method we used earlier for the worst-case problem, i.e., we begin by solving problem $(\textrm{F}_{1})$. Then, from the solution set $\N_{1}$, we find optimal solutions for the problem $(\textrm{F}_{2})$, and so on.

\subsection{Sparsity Level $k=1$}

Assume that the signal $\tb{s}$ is 1-sparse. So, there are ${N \choose 1}=N$ different possibilities to build the matrix $\C_{T_{1}}$ from the matrix $\C$. The expectation in problem $(\textrm{F}_{1})$ can be written as:
\begin{equation}\label{eqa5}
\E_{T_{1}} \lambda_{\textrm{min}}(\C_{T_{1}}^H\C_{T_{1}})=
\sum_{t\in \Omega_{1}}p_{1}(t)\lambda_{\textrm{min}}(\C_{t}^H\C_{t})=\sum_{i=1}^{N}p_{1}({\{i\}})\|\tb{c}_{i}\|^{2}=\frac{m}{N}.
\end{equation}
The following result holds.
\begin{thm}\label{thma1}
The optimal value of the objective function of problem $(\textrm{F}_{1})$ is $m/N$. This value is obtained by using any right orthogonal matrix $\C \in \N_{0}$, i.e., any \emph{tight frame}.
\end{thm}

\begin{proof}The first part is already proved. The proof for optimality is very similar to the proof given in Theorem~\ref{thm1}. Thus, $\N_{1}=\N_{0}$.
\end{proof}

Theorem~\ref{thma1} shows that unlike the worst-case problem, any tight frame is an optimal solution for the problem $(\textrm{F}_{1})$.

Next, we study the case where the signal $\tb{s}$ is 2-sparse.

\subsection{Sparsity Level $k=2$}

For problem $(\textrm{F}_{2})$, the expected value term $\E_{T_{2}} \lambda_{\textrm{min}}(\C_{T_{2}}^H\C_{T_{2}})$ is equal to
\begin{equation}
\sum_{t\in \Omega_{2}}p_{2}(t)\lambda_{\textrm{min}}(\C_{t}^H\C_{t})=
\quad \sum_{j=2}^{N}\sum_{i=1}^{j-1}p_{2}(\{i,j\})\lambda_{\textrm{min}}(\C_{\{i,j\}}^H\C_{\{i,j\}}).\nonumber
\end{equation}
Now, since $p_{2}(t)=1/{N \choose 2}=2/(N(N-1)), \forall t\in \Omega_{2}$, we can go further and write $\E_{T_{2}} \lambda_{\textrm{min}}(\C_{T_{2}}^H\C_{T_{2}})$ as
\begin{equation}\label{neq1}
\frac{2}{N(N-1)}\sum_{j=2}^{N}\sum_{i=1}^{j-1}\lambda_{\textrm{min}}(\C_{\{i,j\}}^H\C_{\{i,j\}}).
\end{equation}
Solving problem $(\textrm{F}_{2})$ with this objective function is not trivial in general. In fact, claiming anything about solutions of the family of problems $(\textrm{F}_{k})$, $k=2,3,\ldots,$ is hard. However, if we constrain ourselves to the class of \emph{uniform} tight frames, which also arise in solving the worst-case problem, we can establish necessary and sufficient conditions for optimality. Nonetheless, these conditions are different from those for the worst-case problem and as we will show next the optimal solution here is a uniform tight frame for which a cumulative measure of coherence is minimal.

Let $\M_{1}$ be defined as $\M_{1}=\{\C:\C \in \N_{1}, \| \tb{c}_{i} \|=\sqrt{m/N},\forall i \in \Omega \}$. Also, for $k=2,3,\ldots,$ recursively define the set $\M_{k}$ as the solution set to the following optimization problem:
\begin{equation}\label{neq2}
(\textrm{F}^{'}_{k}) \quad
\left\{\begin{array}{cl}
\mathop{\max}\limits_{\C} & \E_{T_{k}}  \lambda_{\textrm{min}}(\C_{T_{k}}^{H}\C_{T_{k}}),\\
\textrm{s.t.} & \C \in \M_{k-1}.
\end{array}\right.
\end{equation}
We will concentrate on solving the above problems instead of the family of problems $(\textrm{F}_{k})$, $k=2,3,\ldots$. For $k=2$, we have the following result.
\begin{thm}\label{thma2}
The matrix $\C$ is in $\M_{2}$ if and only if the frame sum-coherence $\sum_{j=2}^{N}\sum_{i=1}^{j-1}|\langle\tb{c}_{i},\tb{c}_{j}\rangle|$ is minimized.
\end{thm}
\begin{proof}
For $k=2$, the value of $\lambda_{\textrm{min}}(\C_{t}^H\C_{t})$ for $t=\{i,j\}\in \Omega_{2}$ is equal to
\[
\lambda_{\textrm{min}}(\C_{\{i,j\}}^H\C_{\{i,j\}})=(1/2)(\|\tb{c}_{i}\|^{2}+\|\tb{c}_{j}\|^{2}-f(i,j)),
\]
where $f(i,j)$ is defined as $f(i,j)=\sqrt{(\|\tb{c}_{i}\|^{2}-\|\tb{c}_{j}\|^{2})^{2}+4\langle\tb{c}_{i},\tb{c}_{j}\rangle^2}$. Now, if we replace this in~\eqref{neq1}, we get
\begin{align*}
&\frac{1}{N(N-1)}\left(\sum_{j=2}^{N}\sum_{i=1}^{j-1}\|\tb{c}_{i}\|^{2}+\|\tb{c}_{j}\|^{2}-f(i,j)\right)\\
&=\frac{1}{N(N-1)}\left((N-1)\sum_{i=1}^{N}\|\tb{c}_{i}\|^{2}-\sum_{j=2}^{N}\sum_{i=1}^{j-1}f(i,j)\right)\\
&=\frac{(N-1)m}{N(N-1)}-\frac{1}{N(N-1)}\sum_{j=2}^{N}\sum_{i=1}^{j-1}f(i,j)=\frac{m}{N}-\frac{1}{N(N-1)}\sum_{j=2}^{N}\sum_{i=1}^{j-1}f(i,j).
\end{align*}
Since $\C\in \M_{1}$, then using the fact that $\|\tb{c}_{i}\|=\sqrt{m/N}$, $\forall i \in \Omega$, we can go one step further and write the above objective function as
\[
\frac{m}{N}-\frac{2}{N(N-1)}\sum_{j=2}^{N}\sum_{i=1}^{j-1}|\langle\tb{c}_{i},\tb{c}_{j}\rangle|.
\]
Therefore, solving problem $(\textrm{F}^{'}_{2})$ becomes equivalent to solving the following optimization problem:
\begin{equation}\label{eqa6}
\begin{array}{cl}
\mathop{\min}\limits_{\C} & \sum_{j=2}^{N}\sum_{i=1}^{j-1}|\langle \tb{c}_{i},\tb{c}_{j} \rangle|,\\
\textrm{s.t.} & \C \in \M_{1}.
\end{array}
\end{equation}
\end{proof}

Theorem~\ref{thma2} shows that for problem $(\textrm{F}^{'}_{2})$, angles between column pairs of the uniform tight frame $\C$ should be designed in a different way than for the worst-case problem. Several articles (though not many) discuss such frames. In~\cite{Donoho2}, the authors introduce a similar concept where instead of finding the minimum of the above summation, they are looking for the maximum, and call it the ``cluster coherence'' of the frame. In~\cite{Bodmann}, where the authors use frames in coding theory applications, it is proved that the solution to one of the problems discussed in the paper is found by solving~\eqref{eqa6}. However, to the best of our knowledge, finding such a frame system is still an open problem---there is no known general solution for problem~\eqref{eqa6}. We call the value of the optimal objective function of~\eqref{eqa6} the \emph{minimum sum-coherence}. The following lemma provides bounds for the objective function of this optimization problem.
\begin{lem}
For a uniform tight frame $\C$ with column norms equal to $\sqrt{m/N}$, the following inequalities hold:
\[
ab|(N/m-1)-2(N-1)\mu_{\C}^{2}|\leq \sum_{j=2}^{N}\sum_{i=1}^{j-1}|\langle\tb{c}_{i},\tb{c}_{j}\rangle|\leq ab(N-1)\mu_{\C}^{2},
\]
where
\[
a=\left(\frac{(m/N)^2}{1-2(m/N)}\right),\quad b=\left(\frac{N(N-2)}{2}\right).
\]
\end{lem}
\vspace{.2cm}

{\em Proof.} See Appendix A.

\subsection{Sparsity Level $k>2$}

Similar to the worst-case problem, solving problems $(\textrm{F}^{'}_{k})$ for $k>2$ is not only a hard task but also it is not known how to construct  frames with the required properties in practice. This is because the solution sets for these problems all lie in $\M_{2}$ and the problem  $(\textrm{F}^{'}_{2})$ is still an open problem. The following lemma provides a lower bound for the optimal objective function of problem  $(\textrm{F}^{'}_{k})$.
\begin{lem}
The optimal value of the objective function for problem $(\textrm{F}^{'}_{k})$ is bounded below by $(m/N)(1-(k(k-1)/2)\mu_{C})$.
\end{lem}

{\em Proof.} See Appendix B.


\section{Simulation Results}\label{sec5}

As mentioned earlier, constructing uniform tight frames with coherence ${\mu^*}$ is an open problem for arbitrary $(m,N)$ pairs. However, examples of such frames are available for modest values of $m$ and $N$, mostly for $1\le m\le 16$ and $1\le N \le 50$ (see~\cite{webad}). To be more precise, the examples in \cite{webad} are the best uniform tight frames (in terms of coherence) that the site publisher is aware of. In some cases, these frames in fact have coherence $\mu^*$. In other cases, their coherence is larger than $\mu^\ast$. For the minimum sum-coherence problem, the examples are even more scarce, and in fact we are not aware of any examples for $(m,M)$ dimensions large enough to be of interest to our study. Therefore, we limit our numerical study to the worst-case problem, where we evaluate the performance of several uniform tight frames from \cite{webad}.

In all simulations, we assume $\sigma^{2}_{n}=1$ and $\|{\bm \theta}_{T}\|=1$. We present plots of the worst-case $\textrm{SNR}/N$, where the worst-case SNR is given by
\begin{equation}
 \textrm{SNR}=\mathop{\min}\limits_{T}\mathop{\min}\limits_{\theta_{T}}\|{\bm \Phi}{\bm \Psi}_{T}{\bm \theta}_{T}\|^{2}/(\sigma^{2}_{n}/N)=\mathop{\min}\limits_{T}\lambda_{min}(\C^{*H}_{T}\C^{*}_{T}),\nonumber
\end{equation}
by fixing two of the three variables $m$, $N$, and $k$ and changing the third one.
\begin{figure}[!th]
\begin{center}
\begin{tabular}{cc}
\includegraphics[width=62mm,height=4.6cm]{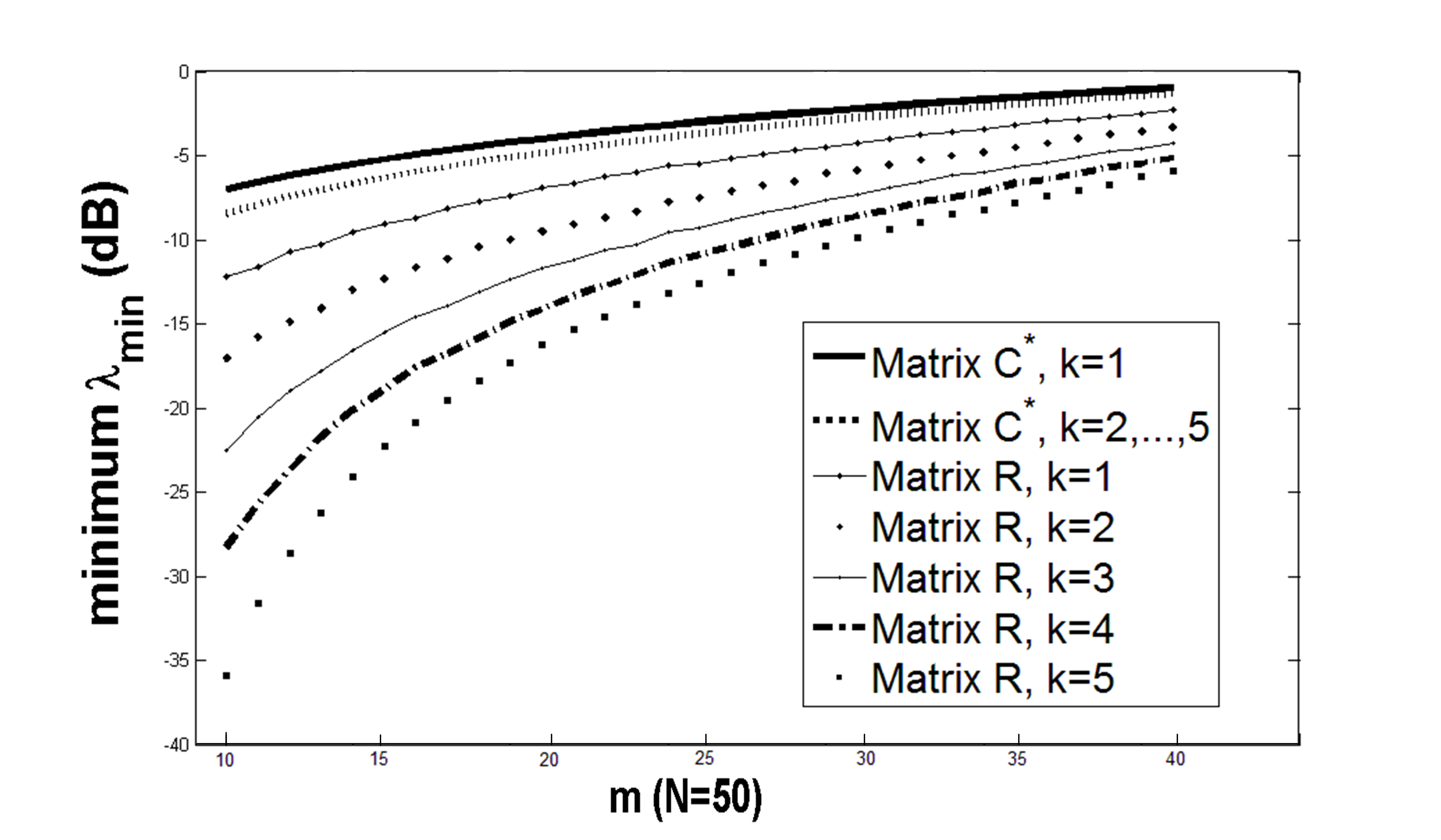} &
\includegraphics[width=62mm,height=4.6cm]{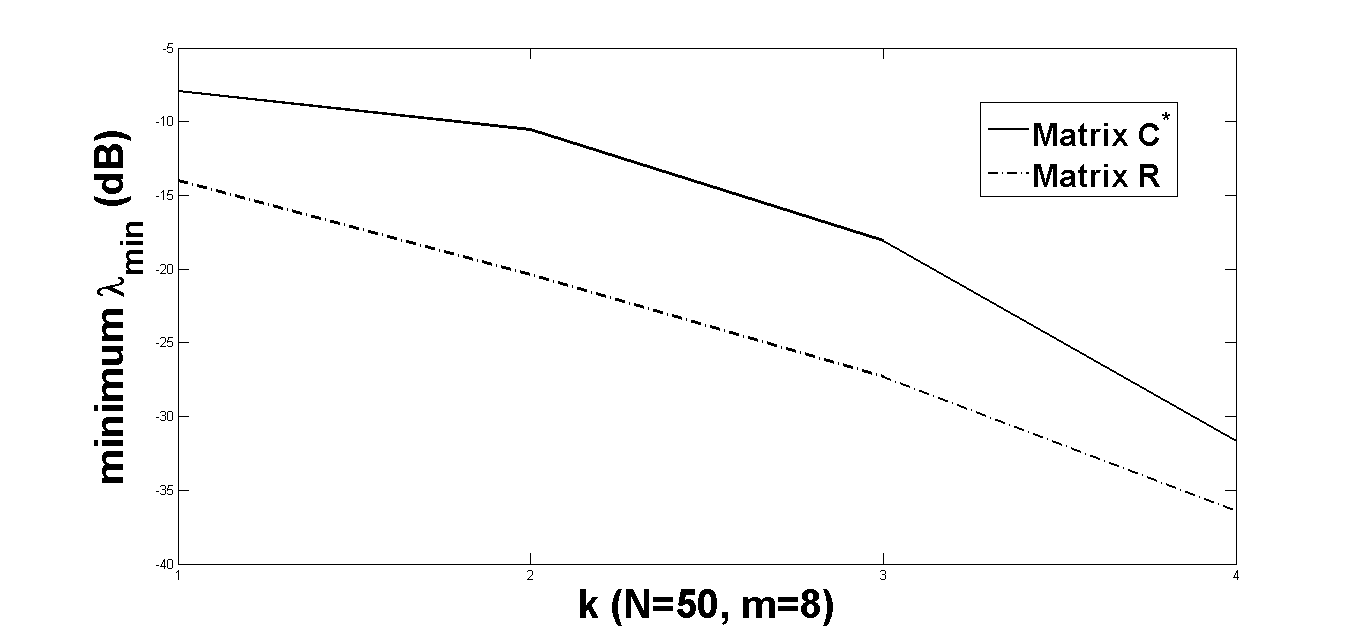}\\
(a) & (b)\\
\includegraphics[width=62mm,height=4.6cm]{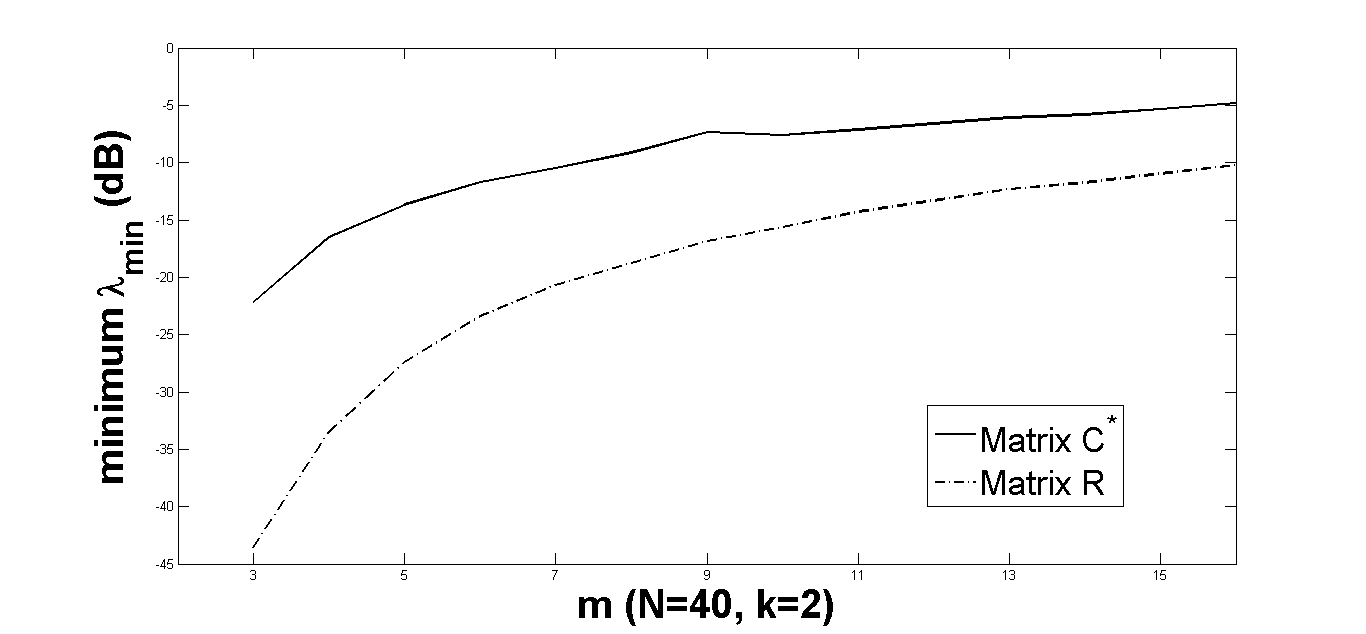}&
\includegraphics[width=62mm,height=4.6cm]{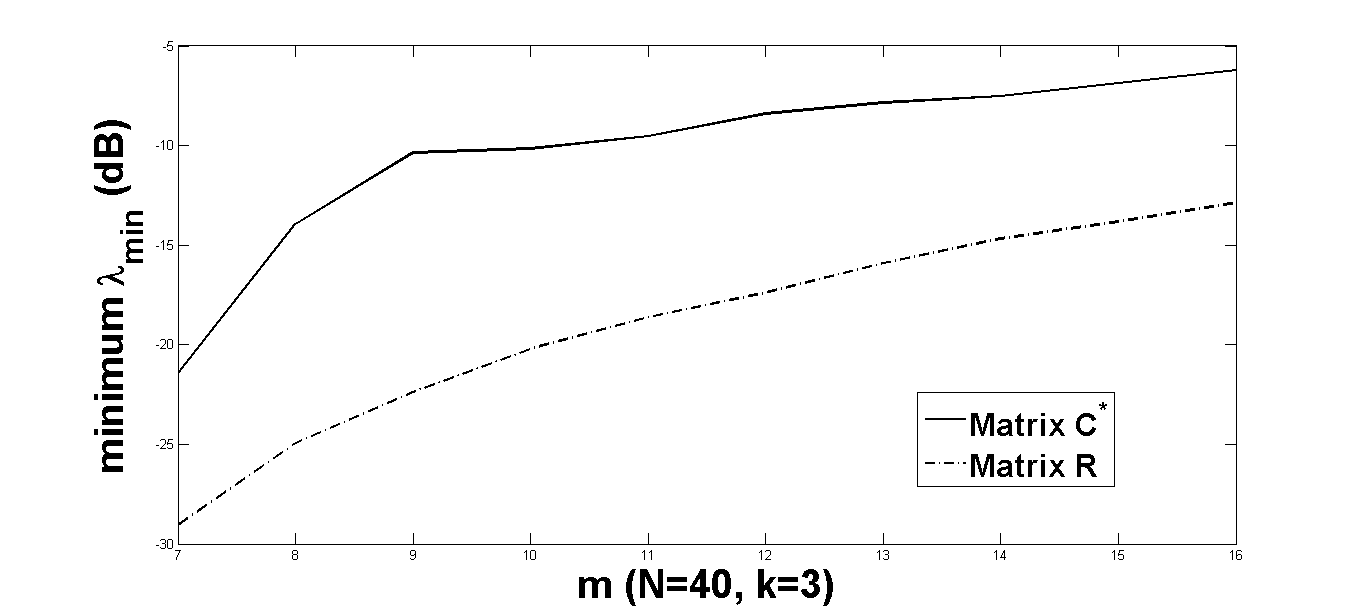}\\
(c) & (d)\\
\includegraphics[width=62mm,height=4.6cm]{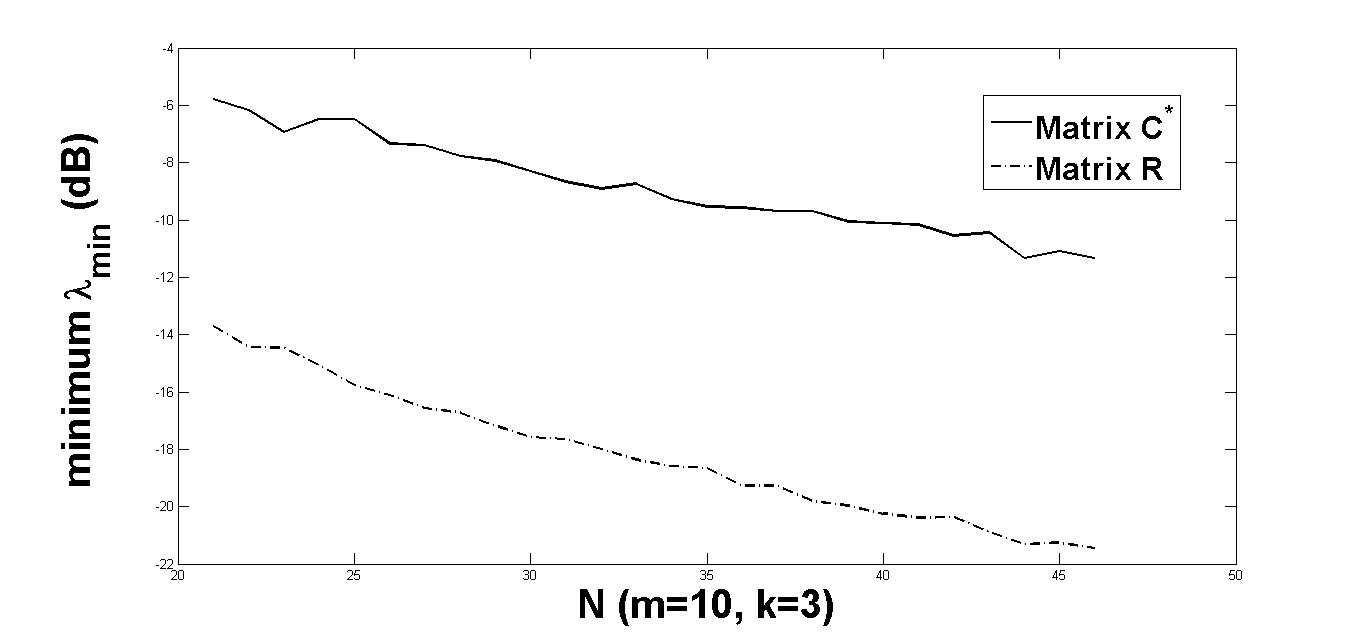}&
\includegraphics[width=62mm,height=4.6cm]{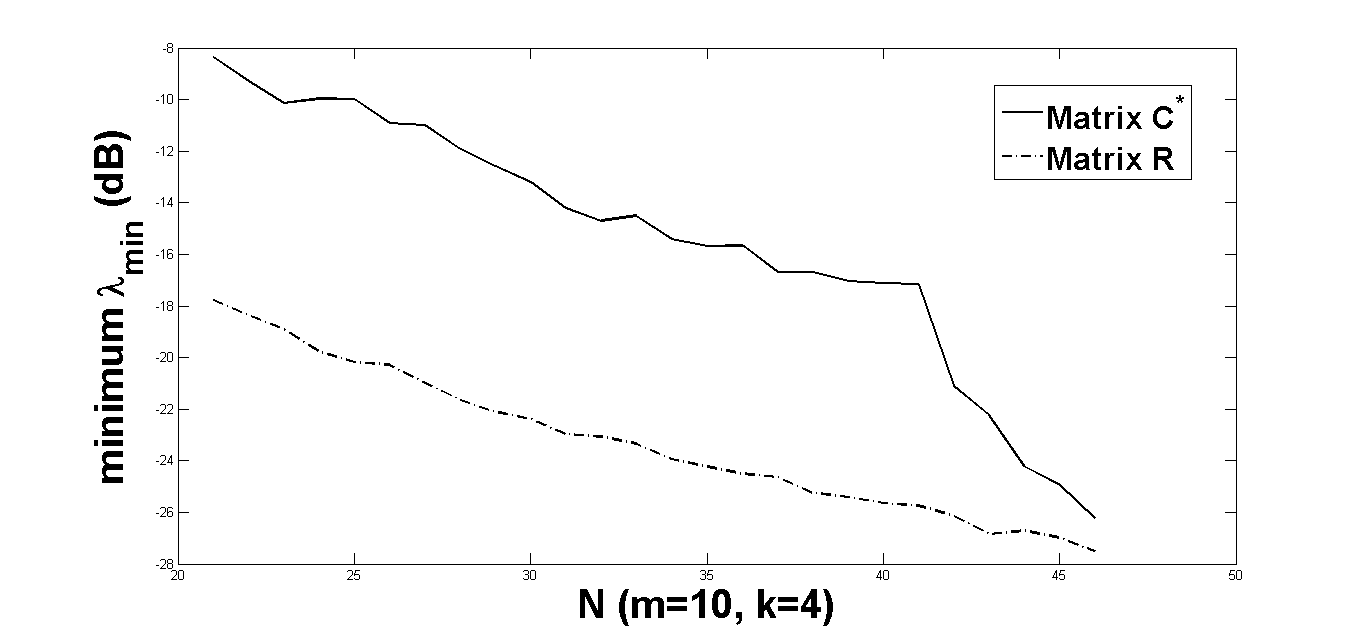}\\
(e) & (f)
\end{tabular}
\end{center}
\caption{Performance comparison between matrices $\C^{*}$ and $\mathbf{R}$: (a) equiangular uniform tight frames. (b) $m=8$, $N=50$, and $k$ is varied. (c) $k=2$, $N=40$, and $m$ is varied. (d) $k=3$, $N=40$, and $m$ is varied (e) $m=10$, $k=3$, and $N$ is varied. (f) $m=10$, $k=4$, and $N$ is varies.}\label{f:sim}
\end{figure}
We compare the performance of our robust (worst-case) design $\C^{*}$ with that of a matrix $\mathbf{R}$ with i.i.d.\ Gaussian
$\mathcal{N}(0,(1/m))$ entries, which is typically used for signal recovery. To satisfy the constraint in problem~\eqref{eq1}, we make $\mathbf{R}$ to be right orthogonal. The value of the objective function in~(\ref{eq1}) is averaged over $100$ realizations of the matrix $\mathbf{R}$.

Figure \ref{f:sim}(a) shows the worst-case SNR performance for a case where the signal dimension is $N=50$ and the measurement budget $m$ is varied from $10$ to $40$. In this case, the condition~(\ref{eq24}) is satisfied and the columns of the optimal matrix $\C^{*}$ form an equiangular uniform tight frame. We can therefore derive an exact expression for the optimal objective function value based on the Welch bound. For $k=1$, this value is equal to $m/N$, and for $k \geq2$, it is equal to $(m/N)(1-\mu^{*})$ where $\mu^{*}=\sqrt{(N-m)/(m(N-1))}$.

We also consider cases where the condition~(\ref{eq24}) is not satisfied, due to a relatively small measurement budget. Here we use Grassmannian line packings to form measurement matrices. For  $(N,m)$ pairs that Grassmannian line packings are not known, we use the best available packings reported in \cite{webad} for those dimensions. Figures \ref{f:sim}(b)-(f) show the performance of such solutions versus the random matrix $\R$ for different case. In each case, we have fixed two of the variables $N$, $m$, and $k$ and have varied the third one. The values of the objective functions in all these plots are in dB. In all scenarios, the wort-case SNR performance corresponding to the the optimal design $\C^{*}$ is better than the average taken over 100 realization the random matrix $\R$.

Note that our simulations are only for cases where $m$, $k$ and $N$ are not very large. As mentioned above, one of the reasons is that the available uniform tight frames in~\cite{webad} are mostly for cases where $1\leq m \leq 16$ and $1\leq N \leq 50$. Also, for values of $N$ bigger than 25 and $k$ bigger than 5, finding the smallest minimum eigenvalue of all ${{\C}_T^*}^H {\C}_T^*$ for different values of $T$ is computationally intractable.

It is important to realize that for most values of $m$ and $N$, the uniform tight frames used in our simulations have a coherence $\mu$ that is bigger than $\mu^*$. In other words, for most values of $m$ and $N$, we are actually comparing the performance of a suboptimal solution matrix instead of the optimal solution with the performance of the random matrix $\R$ and interestingly, the suboptimal solution still has a better performance than the random matrix $\R$ in most, but not all, cases. For example, we notice that in Figure 1(f), the gap between two curves decreases as $N$ increases. This does not contradict with our theoretical results, as the plots in Figure \ref{f:sim} do not show the performance of the optimal solution for most values of $m$ and $N$ after all. Rather they show the performance of the best available uniform tight frame example for the corresponding ($m,N$) values.


\section{Conclusions}\label{sec6}

In this paper, we have considered the design of low-dimensional (compressive) measurement matrices, for a given number of measurements, for maximizing the worst-case SNR and the average minimum SNR. We have shown an interesting connection between maximizing the two SNR criteria for detection and certain classes of frames. In the worst-case SNR problem, we have shown that the optimal measurement matrix is a Grassmannian line packing for most---and a uniform tight frame for all---sparse signals. In the average SNR problem, we have looked for the solution among the class of uniform tight frames and have shown that the optimal measurement matrix is a uniform tight frame that has minimum sum-coherence. Our solutions for both problems provide lower bounds for the performance of the detectors.

\section*{Appendix}
\subsection*{A. Proof of Lemma 1.}
Multiply both sides of $\C\C^H=\I$ from the left by $\C^H$ and from the right side by $\C$ to get
\begin{equation}\label{eqa7}
(\C^H\C)^{2}=\C^H\C.
\end{equation}
The matrix $\C^H\C=\I$ is an ($N \times N$) Hermitian matrix, with $(i,j)$th element $(m/N)\cos\alpha_{ij}$ and diagonal elements $m/N$. Using these values, it is easy to see that the matrix $(\C^H\C)^{2}$ is also a Hermitian matrix with the entry $(m/N)^{2}(\sum_{i=1}^{N}\cos^{2}{\alpha_{ji}})$ on the $j$th diagonal location and the entry
\[
(m/N)^{2}(2\cos{\alpha_{ij}}+\sum_{l=1,\atop l\neq i,j}^{N}\cos{\alpha_{il}}\cos{\alpha_{lj}})
\]
located in the $i$th row and the $j$th column. By comparing the diagonal entries on each side of equation~\eqref{eqa7}, we will get the following family of equations:
\[
\left(\frac{m}{N}\right)^{2}\left(\sum_{i=1}^{N}\cos^{2}{\alpha_{ji}}\right)=\left(\frac{m}{N}\right), \quad j=1,\ldots,N.
\]
If we sum up all the above equations, after simplifying, we get\footnote{The relation \eqref{eqj1} is the well-known frame potential condition (see \cite{benedetto2003finite})
\[
\textrm{FP}=\sum_{i,j=1}^{N}|\langle \tb{c}_{i},\tb{c}_{j} \rangle|^2=m
\]
for tight frames, after it has been simplified by enforcing the equal norm assumption.}
\begin{equation}\label{eqj1}
\sum_{i,j=1}^{N}\cos^{2}{\alpha_{ji}}=\frac{N^{2}}{m}.
\end{equation}
If we compare the off-diagonal entries of matrices on each side of equation~\eqref{eqa7}, then for $i,j=1,\ldots,N$ and $i\neq j$, we get
\[
\left(\frac{m}{N}\right)^{2}\left(2\cos{\alpha_{ij}}+\sum_{l=1,\atop l\neq i,j}^{N}\cos{\alpha_{il}}\cos{\alpha_{lj}}\right)=\left(\frac{m}{N}\right)\cos{\alpha_{ij}},
\]
which simplifies to
\[
\cos{\alpha_{ij}}=\left(\frac{(m/N)}{1-2(m/N)}\right)\left(\sum_{l=1,\atop l\neq i,j}^{N}\cos{\alpha_{il}}\cos{\alpha_{lj}}\right).
\]
Using the triangle inequality, we write
\begin{align*}
\sum_{j=2}^{N}\sum_{i=1}^{j-1}|\langle \tb{c}_{i},\tb{c}_{j} \rangle|&=\frac{m}{N}\sum_{j=2}^{N}\sum_{i=1}^{j-1}|\cos{\alpha_{ij}}|\\
&\geq \frac{m}{N}\left| \sum_{j=2}^{N}\sum_{i=1}^{j-1} \cos{\alpha_{ij}}\right| \\
&=a\left| \sum_{j=2}^{N}\sum_{i=1}^{j-1} \sum_{l=1,\atop l\neq i,j}^{N}\cos{\alpha_{il}}\cos{\alpha_{lj}}\right|.
\end{align*}
We replace $\cos{\alpha_{il}}\cos{\alpha_{lj}}$ with $(1/2)(\cos^{2}{\alpha_{il}}+\cos^{2}{\alpha_{lj}}-(\cos{\alpha_{il}}-\cos{\alpha_{lj}})^{2})$.  The term $\cos^{2}{\alpha_{il}}$ is repeated $2(N-2)$ times in the above summation;
\begin{itemize}
\item Once $i$ and $l$ are fixed, there are $N-2$ choices left for $j$ to choose the angle $\alpha_{lj}$ in the product term $\cos{\alpha_{il}}\cos{\alpha_{lj}}$.
 \item There are also $N-2$ times that the term $\cos{\alpha_{jl}}$ is repeated, which is equal to $\cos{\alpha_{lj}}$.
\end{itemize}
Therefore,
\begin{align*}
\sum_{j=2}^{N}\sum_{i=1}^{j-1} \sum_{l=1,\atop l\neq i,j}^{N}\cos^{2}{\alpha_{il}}+\cos^{2}{\alpha_{lj}}&=2(N-2)\sum_{j=2}^{N}\sum_{i=1}^{j-1}\cos^{2}{\alpha_{ij}}\\
&=2(N-2)\frac{(N^{2}/m)-N}{2}.
\end{align*}
The right hand side of the above inequality simplifies to
\[
(a/2)\left|N(N-2)(N/m-1)-\sum_{j=2}^{N}\sum_{i=1}^{j-1} \sum_{l=1,\atop l\neq i,j}^{N}(\cos{\alpha_{il}}-\cos{\alpha_{lj}})^{2}\right|.
\]
It is easy to show that $|\cos{\alpha_{il}}-\cos{\alpha_{lj}}|\leq 2\mu_{\C}$ for any $i\neq j \neq l=1,\ldots,N$. So,
\[
-\sum_{j=2}^{N}\sum_{i=1}^{j-1} \sum_{l=1,\atop l\neq i,j}^{N}(\cos{\alpha_{il}}-\cos{\alpha_{lj}})^{2}\geq -4\sum_{j=2}^{N}\sum_{i=1}^{j-1} \sum_{l=1,\atop l\neq i,j}^{N}\mu_{\C}^{2}.
\]
Similarly, for a fixed $i$ and $j$, there are $N-2$ possibilities for $l$. Also, there are ${N \choose 2}$ ways to choose $i$ and $j$ from $N$ options. Therefore, the lower bound will be larger than
\begin{align*}
&(a/2)\left|N(N-2)(N/m-1)-2N(N-1)(N-2)\mu_{\C}^{2}\right|\\
&=ab|(N/m-1)-2(N-1)\mu_{\C}^{2}|.
\end{align*}
This is the claimed lower bound.

To find the upper bound, we write
\begin{align*}
\sum_{j=2}^{N}\sum_{i=1}^{j-1}|\langle \tb{c}_{i},\tb{c}_{j} \rangle|&=\frac{m}{N}\sum_{j=2}^{N}\sum_{i=1}^{j-1}|\cos{\alpha_{ij}}|\\
&= a\sum_{j=2}^{N}\sum_{i=1}^{j-1} \left| \sum_{l=1,\atop l\neq i,j}^{N}\cos{\alpha_{il}}\cos{\alpha_{lj}}\right| \\
&\leq a\sum_{j=2}^{N}\sum_{i=1}^{j-1} \sum_{l=1,\atop l\neq i,j}^{N}\left|\cos{\alpha_{il}}\cos{\alpha_{lj}}\right| \\
&\leq a\sum_{j=2}^{N}\sum_{i=1}^{j-1} \sum_{l=1,\atop l\neq i,j}^{N}\mu_{\C}^{2}\\
&=ab(N-1)\mu_{\C}^{2}.
\end{align*}\hfill $\Box$

\subsection*{B. Proof of Lemma 2.}
Similar to the 2-sparse signals case in Section \ref{sec2}, we can write the objective function of problem $(\textrm{F}^{'}_{k})$ in the following way:
\[
\E_{T_{k}} \lambda_{\textrm{min}}(\C_{T_{k}}^H\C_{T_{k}})=\sum_{t\in \Omega_{k}} p_{k}(t)\lambda_{\textrm{min}}(\C_{t}^H\C_{t})
={N \choose k}^{-1}\sum_{t\in \Omega_{k}} \lambda_{\textrm{min}}(\C_{t}^H\C_{t}).
\]
Since the matrix $\C$ is a uniform tight frame, for any $t \in \Omega_{k}$, the matrix $\C_{t}^H\C_{t}$ can be written as $(m/N)[\I+\mathbf{A}_{t}]$, where the $(k \times k)$ matrix $\mathbf{A}_{t}$ is defined in~\eqref{eq29}. Similar to the worst-case design, we can derive the following inequality:
\begin{align*}
\lambda_{\textrm{min}}(\C_{t}^H\C_{t})& \geq (\frac{m}{N})(1-\sum_{i_{l}\neq i_{h} \in t, \atop t \in \Omega_{k}}|\cos{\alpha_{i_{l}i_{h}}}|)\\
&\geq (\frac{m}{N})(1-{k \choose 2}\mu_{\C})\\
&=(\frac{m}{N})(1-(k(k-1)/2)\mu_{\C}).
\end{align*}
Taking the expectation, we get
\begin{align*}
\E_{T_{k}} \lambda_{\textrm{min}}(\C_{T_{k}}^H\C_{T_{k}})&\geq p\sum_{t\in \Omega_{k}}(\frac{m}{N})(1-(k(k-1)/2)\mu_{\C})\\
&=p{N \choose k}(\frac{m}{N})(1-(k(k-1)/2)\mu_{\C})\\
&=(\frac{m}{N})(1-(k(k-1)/2)\mu_{\C}).
\end{align*}
This completes the proof.\hfill $\Box$



\end{document}